\documentclass[version=preprint]{iacrcc}

\license{CC-by}

\title[running  = {A Survey of Interactive Verifiable Computing},
       subtitle = {Utilizing Low-degree Polynomials}
      ]{A Survey of Interactive Verifiable Computing}

\addauthor[orcid    = {0009-0004-9931-6589},
           inst     = {1},
           onclick  = {https://org.weids.dev},
           footnote = {Hong Kong Baptist University},
           email    = {f3459026@comp.hkbu.edu.hk},
           surname  = {Jiawei}]{Angold Wang}

\addaffiliation[ror     = 0145fw131,
                street  = {224 Waterloo Rd},
                department = {Department of Computer Science},
                city    = {Kowloon Tong},
                country = {Hong Kong}
                ]{Hong Kong Baptist University}

\theoremstyle{plain}
\theoremstyle{definition}
\usepackage{algorithm}
\usepackage{caption}

\begin{document}

\maketitle


\begin{abstract}
This survey provides a comprehensive examination of verifiable computing, tracing its evolution from foundational complexity theory to modern zero-knowledge succinct non-interactive arguments of knowledge (ZK-SNARKs). We explore key developments in interactive proof systems, knowledge complexity, and the application of low-degree polynomials in error detection and verification protocols. The survey delves into essential mathematical frameworks such as the Cook-Levin Theorem, the sum-check protocol, and the GKR protocol, highlighting their roles in enhancing verification efficiency and soundness. By systematically addressing the limitations of traditional NP-based proof systems and then introducing advanced interactive proof mechanisms to overcome them, this work offers an accessible step-by-step introduction for newcomers while providing detailed mathematical analyses for researchers. Ultimately, we synthesize these concepts to elucidate the GKR protocol, which serves as a foundation for contemporary verifiable computing models. This survey not only reviews the historical and theoretical advancements in verifiable computing over the past three decades but also lays the groundwork for understanding recent innovations in the field.
\end{abstract}

\begin{textabstract}
This survey provides a comprehensive examination of zero-knowledge interactive verifiable computing, emphasizing the utilization of randomnes in low-degree polynomials. We begin by tracing the evolution of general-purpose verifiable computing, starting with the foundational concepts of complexity theory developed in the 1980s, including classes such as P, \(\mathsf{NP}\) and NP-completeness. Through an exploration of the Cook-Levin Theorem and the transformation between \(\mathsf{NP}\) problems like \textsc{HAMPATH} and SAT, we demonstrate the reducibility of \(\mathsf{NP}\) problems to a unified framework, laying the groundwork for subsequent advancements.

Recognizing the limitations of NP-based proof systems in effectively verifying certain problems, we then delve into interactive proof systems (IPS) as a probabilistic extension of NP. IPS enhance verification efficiency by incorporating randomness and interaction, while accepting a small chance of error for that speed. We address the practical challenges of traditional IPS, where the assumption of a \(P\) with unlimited computational power is unrealistic, and introduce the concept of secret knowledge. This approach allows a \(P\) with bounded computational resources to convincingly demonstrate possession of secret knowledge to 

\(V\), thereby enabling high-probability verification 
\(V\). We quantify this knowledge by assessing \(V\)’s ability to distinguish between a simulator and genuine \(P\), referencing seminal works such as Goldwasser et al.’s "The knowledge Complexity of Theorem Proving Procedures"

The survey further explores essential mathematical theories and cryptographic protocols, including the Schwartz-Zippel lemma and Reed-Solomon error correction, which underpin the power of low-degree polynomials in error detection and interactive proof systems. We provide a detailed, step-by-step introduction to tyhe sum-check protocol, proving its soundness and runtime characteristics.

Despite the sum-check protocol’s theoretical applicability to all \(\mathsf{NP}\) problems via SAT reduction, we highlight the sum-check protocol’s limitation in requiring superpolynomial time for general-purpose computations of a honest \(P\). To address these limitations, we introduce the GKR protocol, a sophisticate general-purpose interactive proof system developed in the 2010s. We demonstrate how the sum-check protocol integrates into the GKR framework to achieve efficient, sound verification of computations in polynomial time. This survey not only reviews the historical and theoretical advancement in verifiable computing in the past 30 years but also offers an accessible introduction for newcomers by providing a solid foundation to understand the significant advancements in verifiable computing over the past decade, including developments such as ZK-SNARKs.
\end{textabstract}

\section{Introduction}

Verifiable computing has been a pivotal area of research since the latter part of the 20th century, rooted in the study of computational complexity~\cite{DBLP:conf/stoc/Cook71, DBLP:journals/corr/abs-1009-5894}. Early investigations focused on understanding the inherent difficulty of various computational problems and developing abstractions that illustrate the hardness of these problems and the methods to solve them efficiently~\cite{DBLP:books/fm/GareyJ79}. As computing power surged, the focus shifted towards more robust computing models that enable resource-constrained clients to delegate complex computations to powerful servers. Crucially, these models ensure that the server can provide a proof of correctness for the delegated tasks, embodying the essence of verifiable computing~\cite{DBLP:conf/stoc/GoldwasserMR85}.

In recent years, the field has witnessed significant advancements~\cite{DBLP:conf/stoc/GoldwasserKR08, DBLP:conf/sp/Ben-SassonCG0MTV14}, particularly in enhancing the efficiency of verifiable computing protocols on both the client and server sides~\cite{DBLP:conf/sp/ParnoHG013, DBLP:conf/eurocrypt/Groth16, 8418611, cryptoeprint:2019/953}, these advancements often involve sophisticated mathematical concepts from cryptography. This survey aims to chart the historical development of verifiable computing, starting with fundamental concepts and models, identifying their limitations, and progressing towards more powerful and efficient protocols. To ensure accessibility, each concept introduced is accompanied by one or two minimal examples with detailed descriptions. Additionally, every claim and analysis is supported by rigorous mathematical proofs, with prerequisite mathematical foundations provided to aid newcomers. The survey is structured as follows:

\paragraph{Section 2: Complexity Theory Foundations}
We begin by exploring the foundational concepts of complexity theory developed in the 1980s, including complexity classes such as P, NP, and NP-completeness. Through the Cook-Levin Theorem~\cite{DBLP:conf/stoc/Cook71, DBLP:journals/corr/abs-1009-5894} and the transformation of \(\mathsf{NP}\) problems like \textsc{HAMPATH} and SAT, we demonstrate the reducibility of \(\mathsf{NP}\) problems to a unified framework, setting the stage for further developments.

\paragraph{Section 3: Interactive Proof Systems (IPS)}
Recognizing the limitations of NP-based proof systems in verifying certain problems effectively, we delve into interactive proof systems. IPS extend \(\mathsf{NP}\) by incorporating randomness and interaction, enhancing verification efficiency at the cost of allowing a small probability of error.
To address the practical challenges inherent in traditional IPS, we explore the concept of knowledge complexity, as originally introduced in~\cite{DBLP:conf/stoc/GoldwasserMR85}, this approach enables a \(P\)\footnote{In this survey, we use \(P\) and \(V\) to represent the prover and the verifier in the interactive proof model.} with bounded computational resources to convincingly demonstrate possession of secret knowledge to \(V\), facilitating high-probability verification by quantifying \(V\)’s ability to distinguish between a simulator and a genuine \(P\), knowledge complexity elucidates how zero-knowledge protocols maintain the secrecy of underlying information while still proving its validity.

\paragraph{Section 4: The Power of Low-Degree Polynomials}
This section explores essential mathematical theories and cryptographic protocols based on Thaler's seminar PPT~\cite{thalerppt}, including the Schwartz-Zippel lemma and Reed-Solomon error correction. These underpin the efficacy of low-degree polynomials in error detection and interactive proof systems.

\paragraph{Section 5: The Sum-Check Protocol}
We provide a detailed, step-by-step introduction to the sum-check protocol~\cite{89518}, proving its soundness and analyzing its runtime characteristics. While theoretically applicable to all \(\mathsf{NP}\) problems via SAT reduction, we highlight its limitation in requiring superpolynomial time for general-purpose computations by an honest \(P\).

\paragraph{Section 6: The GKR Protocol}
To overcome the limitations of the sum-check protocol, we introduce the GKR protocol~\cite{DBLP:conf/stoc/GoldwasserKR08}, a sophisticated general-purpose interactive proof system developed in the 2010s. We demonstrate how the sum-check protocol integrates into the GKR framework to achieve efficient and sound verification of computations in polynomial time for both \(P\) and \(V\). This protocol forms the foundation for the most recent verifiable computing models, including ZK-SNARKs.

\paragraph{}
In conclusion, this survey not only reviews the historical and theoretical advancements in verifiable computing over the past thirty years but also provides an accessible foundation for understanding significant recent developments. By assembling the various components of verifiable computing, we offer a cohesive understanding of the field’s progression towards efficient and secure computational verification protocols.

\newpage


\section{Preliminaries: Complexity Theory}

Complexity theory is a cornerstone of theoretical computer science, focusing on classifying computational problems based on the resources required to solve them, such as time and space. This field seeks to understand the fundamental limits of what can be efficiently computed and to categorize problems according to their inherent difficulty. This section provides a survey of key concepts in complexity theory, drawing upon the foundational texts by Michael Sipser \cite{DBLP:books/daglib/0086373}, M. R. Garey and D. S. Johnson \cite{DBLP:books/fm/GareyJ79}.
\subsection{P (Polynomial Time)}
\label{sec:orgc135b29}
\textbf{P} (Polynomial time decidable languages) is the class of languages that is decidable in polynomial time (\(\mathcal{O}(n^k)\)) on a deterministic single-tape Turing machine. Formally, P is defined as:
\[P = \bigcup_{k \in \mathbb{N}} \text{Time}(n^k)\] where \( k \) is a constant. This definition encompasses all decision problems deemed tractable, meaning their solutions can be efficiently computed as the input size increases.

\subsection{NP (Nondeterministic Polynomial time)}
\label{sec:orgb0392aa}

In P, we can avoid brute-force search and solve the problem using efficient algorithms. However, attempts to avoid brute force in certain other problems, including many interesting and useful ones, haven't been successful, and polynomial time algorithms that solve them aren't known to exist.

\paragraph{Language View:~\cite{sp}}
\begin{itemize}
    \item \textbf{P} is the class of language which membership can be \textbf{decided} (solved) quickly.
    \item \textbf{NP} is The class of language which membership can be \textbf{veriﬁed} quickly.
\end{itemize}

\paragraph{Turing Machine View:}
\begin{itemize}
    \item \textbf{P} is the set of problems \textbf{solvable} in polynomial time by a \textbf{deterministic TM}
    \item \textbf{NP} is the set of problems \textbf{verifiable} in polynomial time by a \textbf{deterministic TM} and \textbf{solvable} in polynomial time by a \textbf{non-deterministic TM}.
\end{itemize}

\subsubsection{NTM Decider}
\label{sec:orgf455590}
A Non-deterministic Turing Machine (NTM) decider is guaranteed to halt on all inputs. \cite{DBLP:books/daglib/0086373}
An NTM decider is designed so that every possible computation branch halts, either by accepting or rejecting the input.
This means that for any input the machine is given, it will always come to a conclusion (halt) within a finite number of steps.
There are no branches where the machine runs forever without deciding the outcome.
\subsubsection{Solving NP}
\label{sec:org96fe9c8}
The non-deterministic nature of \(\mathsf{NP}\) gives us an abstraction to imagine a machine (NTM) that could guess a solution "in parallel" and verify it quickly.
If we had such a machine, it would allow us to "solve" \(\mathsf{NP}\) problems quickly by magically finding the right solution path.
However, for real-world deterministic machines, we still don’t have efficient algorithms to solve many \(\mathsf{NP}\) problems.

\subsubsection{Example: Hamiltonian Path Problem}
\label{sec:org86bf559}
If a directed or undirected graph, G, contains a \textbf{Hamiltonian path}, a path that visits every vertex in the graph exactly once.
The \textsc{HAMPATH} problem has a feature called polynomial verifiability that is important to understand its complexity.
Verifying the existence of a Hamiltonian path may be much easier than determining its existence.

\paragraph{Theorem: \textsc{HAMPATH} belongs to NP:}
\label{sec:org081000b}

On input \(\langle G, s, t \rangle\) (Say G has m nodes):
\label{sec:org82e6224}
we non-deterministically write a sequence \(v_1,v_2,...,v_m\) of m nodes, and only accept iff: a. \(v_1 = s\)
b. \(v_m = t\) 
and c. each \((v_i,v_{i+1})\) is an edge and no \(v_i\) repeats.

\paragraph{We do not know whether the Co-HAMPATH is in NP:}
The reason is we do not know whether or not we can give a short certificate for a graph not have a Hamiltonian path~\cite{sp}.
\begin{itemize}
    \item \textbf{If P equaled NP: } Then we can test in polynomial time whether a graph has a Hamiltonian path by directly solving the problem, which yields a short certificate.
    \item \textbf{If P not equal to NP: }
then co-HAMPATH is not an \(\mathsf{NP}\) problem, since it is not easily verified.
\end{itemize}

\subsection{NP-Completeness}
\label{sec:org204d277}

The Relationship between \(\mathsf{P}\) and \(\mathsf{NP}\) indicates that whether all problems can be solved in polynomial time, typically, without \textbf{searching}. \textbf{NP-completeness} is a cornerstone concept in computational complexity theory, providing a framework for understanding the inherent difficulty of computational problems. Building upon the detailed exposition presented in Michael Sipser's Lecture Notes \cite{snp}, this section delves into the foundational definitions, key theorems, and proof techniques that characterize NP-complete problems.

\subsubsection{P \( \equiv \) NP}
\label{sec:orgb645961}
You can always eliminate searching.
If these classes were equal, any polynomially verifiable problem would be polynomially decidable.
\subsubsection{P \( \neq \) NP}
\label{sec:org3ef18a2}
There were cases where you need to search. \cite{DBLP:books/daglib/0086373}
\textit{"Most researchers believe that the two classes are not equal because people have invested enormous effort to find polynomial time algorithms for certain problems in NP, without success.
Researchers also have tried proving that the classes are unequal, but that would entail showing that no fast algorithm exists to replace brute-force search."}

\paragraph{Defn: B is NP-complete if}:
\begin{enumerate}
    \item B is a member of NP
    \item For all A in NP, A \( \leq p \) B
\end{enumerate}

Every language in \(\mathsf{NP}\) has the polynomial time reduced to a complete language of NP, which means if B is NP-complete and B is in P then P = NP.
One important advance on the P versus \(\mathsf{NP}\) question came in the early 1970s with the work of Stephen Cook and Leonid Levin.\cite{DBLP:conf/stoc/Cook71, DBLP:journals/corr/abs-1009-5894} which shows that the Boolean Satisfactory Problem (SAT) is NP-complete.

\newpage

\noindent \textbf{NP-completeness is a very important complexity property of any question}:
\begin{enumerate}
    \item Showing NP-complete is strong evidence of computational intractability (hard).
    \item Gives a good candidate for proving P \( \neq \) NP.
\end{enumerate}

Michael Sipser in 2020: \cite{quote}
\textit{"Back 20 years ago, I was working very hard to show the composite number problem is not in P.
And then, turns out, composite was in P }(proved by \cite{primes}).
\textit{So it was the wrong to pick the composite number problem, but what NP-complete is guarantees is that:
If you work on a problem, which is NP-complete, you can't pick the wrong problem, because if any problem is in \(\mathsf{NP}\) and not in P, an NP-complete problem is going to be an example of that.
Because if the NP-complete problems in P, everything in \(\mathsf{NP}\) is in P."}

\subsection{The 3SAT Problem}
\label{sec:org419409e}
\subsubsection{Conjunctive Normal Form (CNF)}
\label{sec:org3b7e654}
A boolean formula \(\phi\) is in Conjunctive Normal Form (CNF) if it has the form:
\[
\phi = (x \lor \neg y \lor z) \land (\neg x \lor \neg s \lor z \lor u) \land ... \land (\neg z \lor \neg u)
\]
\begin{itemize}
\item \textbf{Literal:} a variable \(\neg x\) or a negated variable
\label{sec:orgae477f6}
\item \textbf{Clause:} an OR of the literals.
\label{sec:org7a81f21}
\item \textbf{CNF:} an AND of the clauses.
\label{sec:org72a8d1a}
\item \textbf{3CNF:} a CNF with exactly 3 literals in each clause.
\label{sec:orgdd01e20}
\end{itemize}

\subsubsection{SAT}
\label{sec:orga67d26a}
\textbf{Boolean satisfiability problem} (SAT) is the problem of determining if there exists an interpretation that satisfies a given Boolean formula.
In other words, it asks whether the variables of a given Boolean formula can be consistently replaced by the values (TRUE or FALSE) in such a way that the formula evaluates to TRUE.
\subsubsection{3SAT is the satisfatory problem restricted to 3CNF formulas}
\label{sec:orgbcee8b0}

\[
3SAT = \{ \phi \mid \phi \text{ is a satisfiable 3CNF formular} \}
\]

\subsubsection{Theorem: 3SAT \(\le p \) K-CLIQUE}
\label{sec:orgcdc8b1b}
We will show that we can reduce 3SAT to K-CLIQUE in polynomial time by building a model on 3SAT.
And hence to show that K-CLIQUE problem is also NP-complete.

\paragraph{The K-Clique Problem}

A k-clique in a graph is a subset of k nodes all directly connected by edges, the input of k-clique problem is an undirected graph and k. The output is a clique (closed) with k vertices, if one exists.

\[
\text{K-CLIQUE} = \{ \langle G, k \rangle \mid \text{ graph G contains a k-clique} \}
\]

\paragraph{The K-Clique Problem is in NP:}

You can easily verify that a graph has a k-clique by exhibiting the clique.

\newpage

\begin{theorem}[\textbf{3SAT \(\le p \) K-CLIQUE} \cite{snp}]
Given polynomial-time reduction \(f\) that maps \(\phi\) to \(\langle G,k \rangle\) where \(\phi\) is satisfiable if and only if G has a k-clique. Given the structure of a CNF, a satisfying assignment to a CNF formula has \(\ge 1\) true literal in each clause.
\end{theorem}

\begin{definition}[G]\label{G}
G: Assume each literal in the formula is a node in G, where:
\begin{itemize}
    \item The forbidden edges:
    \begin{enumerate}
    \item No edges within a clause.
    \item No edges that go between inconsistent labels (\(a\) and \(\neg a\) )
    \end{enumerate}
    \item G has all non-forbidden edges.
    \begin{itemize}
    \item k is the number of clauses
    \item Other than those forbidden edges, all other edges are connected.
    \end{itemize}  
\end{itemize}
\end{definition}

\noindent Claim: \textbf{\(\phi\) is satisfiable iff \(G\) has k-clique}.
We will show that we can reduce 3SAT to K-CLIQUE in polynomial time by \textbf{constructing a model based on 3SAT}.
\[
\phi = (a \lor b \lor \neg c) \land (\neg a \lor b \lor d) \land (a \lor c \lor \neg e) \land ... \land (\neg x \lor y \lor \neg z)
\]

\begin{proof}
\textbf{If \(\phi\) is satisfiable \( \implies \) \(G\) has a k-clique.} \cite{snp}We begin by taking any satisfying assignment to \(\phi\), pick 1 true literal in each clause. assuming that 3SAT is solvable. Then the corresponding nodes in G are a k-clique:
\begin{enumerate}
    \item \textbf{There are at least k nodes:} Assign any node to 1 means a valid node G.
    \item \textbf{No forbidden edges among them}: based on Definition \ref{G}, those nodes on different clauses have edges connected, and since all the labels came from the same assignment. (a is true then \(\neg a\) is false, we cannot pick the inconsistent nodes in different clauses)
\end{enumerate}
\end{proof}

\begin{proof}
\textbf{If G has a k-clique\( \implies \) it will make \(\phi\) satisfiable.} \cite{snp}Taking any k-clique in G. It must have 1 node in each clause, because when we construct 3SAT from given G, nodes cannot appear in a clique together, since there are k clauses, each clause must have exactly one node to form a k-clique graph.
\end{proof}

\textbf{Setting each corresponding literal TRUE gives a satisfying assignment to \(\phi\).}, the reduction \(f \) is computable in polynomial time, which suggests that: \textbf{If k-clique can be solved in polynomial time, then 3SAT can be solved in polynomial time.} Conversely, a polynomial-time solution to 3SAT implies that all \(\mathsf{NP}\) problems, including K-clique, are in P.

\subsection{The Cook-Levin Theorem}\label{cook-levin}
\label{sec:orga50a635}
Once we have one NP-complete problem, we may obtain others by polynomial time reduction from it, as we've seen in K-CLIQUE.
However, establishing the first NP-complete problem is more difficult. Now we do so by proving that SAT is NP-complete.
In 1971, Stephen Cook states that the \textbf{Boolean satisfiability problem is NP-complete} \cite{DBLP:conf/stoc/Cook71}.
That means any problem in \(\mathsf{NP}\) can be reduced in polynomial time by a deterministic Turing Machine to the boolean satisfiability problem (SAT).

\newpage

\subsubsection{SAT is in NP}
\label{sec:orge6fa5f0}
A nondeterministic polynomial time machine can guess an assignment to a given formula \(\phi\) and accept if the assignment satisfies \(\phi\).
\subsubsection{For each A in NP, we have A \( \leq p \) SAT:}
\label{sec:org28bf6b3}
(Any language in \(\mathsf{NP}\) is polynomial time reducible to SAT).

\paragraph{Proof Idea:}
\label{sec:org827b134}

\cite{DBLP:books/daglib/0086373}Let N be a nondeterministic Turing machine that decides A in \(n^k\) time for some constant k.
\emph{We are trying to proof that for any w belongs to any \(\mathsf{NP}\) problems, there is a polynomial time reduction procedure that can transform that w to \(\phi(\text{SAT})\)}

\paragraph{Key to the Proof:}
\label{sec:org7620931} For any \( w \) belongs to any \(\mathsf{NP}\) problems, it can be determined in polynomial time by a nondeterministic Turing machine N, say the running time is \(n^k\).
Then we can construct a Tableau for N is an \(n^k \times n^k\) table whose rows are the configurations of a branch of the computation of N on input w.
Which represents the computation steps/history of that branch of NTM (N). Based on this Tableau, by carefully define each part of \emph{Phi}:
\[
\phi = \phi_{cell} \land \phi_{start} \land \phi_{move} \land \phi_{accept}
\]
We can show that the construction time is is \(O(n^{2k})\), the size of \(\phi\) is polynomial in n.
Therefore we may easily construct a reduction that produces \(\phi\) in polynomial time from the input w of any \(\mathsf{NP}\) problem.

\section{Interactive Proof Systems}
\label{sec:orgf2c105b}
\subsection{The Limitation of \(\mathsf{NP}\) Proof Systems}
\label{sec:orgf5bb5e4}
Let's recall Stephen Cook and Leonid Levin's influential model definition of NP: \cite{DBLP:conf/stoc/Cook71}
The \(\mathsf{NP}\) proof-system consists of two communicating deterministic Turing machines A and B: respectively, the \emph{\(P\)} and the \emph{\(V\)}.
Where the \emph{\(P\)} is \textbf{exponential-time}, the \emph{\(V\)} is \textbf{polynomial-time}.
They read a common input and interact in a very elementary way.
On input a string \(x\) belonging to an \(\mathsf{NP}\) language L, A computes a string \(y\) (whose length is bounded by a polynomial in the length of \(x\) ) and writes \(y\) on a special tape that B can read.
B then checks that \(f_l(y) = x\) (where \(f_l\) is a polynomial-time computable function) and, if so, it halts and accepts.
\subsubsection{\(f_l(y) = x\)}
\label{sec:org21b269b}
We can understand \(f_l(y) = x\) in this way: "The output (certificate) \(y\) belongs to the input x, where \(f_l()\) is a function that can check that \(y\) in poly-time."
\subsubsection{Formalization vs. Intuition}
\label{sec:org2c32467}
Sometimes formalization cannot entirely capture the inituitive notions.
In the context of \textbf{theorem-proving}:
NP captures a specific form of proving a theorem, where a proof can be "written down" and verified without interacting with \(P\)~\cite[Section 3]{DBLP:conf/stoc/GoldwasserMR85}
The certificates, which is like a formal written proof, and \(V\) just passively checks it.
This is like reading a proof in a book. Once you have the book, there is no back-and-forth to clarify or ask questions about the proof.

\subsubsection{Example: Co-HAMPTATH Problem}
\label{sec:org6e9373e}
As we mentioned in the previous section, we do not know whether the complement of \textsc{HAMPATH} (Co-HAMPATH) is in NP:
\begin{itemize}
\item y is easy to be verified: \textbf{HAMPATH}

For the Hamiltonian Path (HAMPATH) problem, given a solution (i.e., a path), it's easy for a \(V\) to check it in polynomial time.

\(P\) can just present the path (certificates, or \(y\)),
and \(V\) checks whether it's a valid Hamiltonian path (i.e., visits each vertex exactly once and satisfies the graph's edges).

\item y is hard to be verified: \textbf{Co-HAMPATH}

The Co-HAMPATH problem asks whether a graph does not have a Hamiltonian path. Here, proving the non-existence of something becomes far more complex.

If you ask a \(P\) to convince you that no Hamiltonian path exists, the proof isn't as simple as just pointing to something (like a path).
Instead, you'd need to somehow verify all possible paths don't work, which could take \textbf{exponential time}.
\end{itemize}
\subsubsection{Limitation of the \(\mathsf{NP}\) Proof-System}\label{Limitation of the NP Proof-System}
\label{sec:org5fd113e}
In the \(\mathsf{NP}\) model, some problems in  \(\mathsf{NP}\)  (like HAMPATH) are easily verifiable, but \(\mathsf{NP}\) does not capture the complexity of some other problems (often their complements i.e., Co-NP problems).
That's why problems like Co-HAMPATH are much harder to verify using the static  \(\mathsf{NP}\)  model:
\textbf{In our example, there isno easy way for a \(P\) to present a simple "proof" that no Hamiltonian path exists, and for \(V\) to check it efficiently.}

\subsection{The Interactive Proof Systems}
\label{sec:orgee95a38}
In 1985, Goldwasser et al. \cite{DBLP:conf/stoc/GoldwasserMR85} introduced an interactive proof-systems to capture a more general way of communicating a proof.

Much like in computation, BPP \cite[Section 10.2.1: The class BPP, pp. 336–339]{DBLP:books/daglib/0086373} (Bounded-error Probabilistic Polynomial time) algorithms provide a probabilistic analog to P to enhance efficiency while accepting a small chance of error for that speed.
In verification, \textbf{IP (Interactive Proof) systems provide a way to define a probabilistic analog of the class NP}.
IP includes problems not known to be in NP, demonstrating greater verification power due to \emph{randomness} and \emph{interaction}.

\subsubsection{Interactive Pairs of Turing Machines}
\label{sec:orgd055bfc}
\begin{itemize}
\item \textbf{\(P\) (Prover)}

An entity with unlimited computational power, aiming to convince \(V\) the truth of a statement.
\item \textbf{\(V\) (Verifier)}

A probabilistic polynomial-time Turing machine (with a random tape) that interacts with \(P\) to verify the statement's validity.
\end{itemize}
\subsubsection{Interactions}
\label{sec:org3cf5d94}
The interaction consists of multiple rounds where \(P\) and \(V\) exchange messages.
\(V\) uses randomness to generate challenges, and \(P\) responds accordingly.
The key properties of such systems are:
\begin{itemize}
\item \textbf{Completeness:}
If the statement is true, an honest \(P\) can convince \(V\) with high probability.
\item \textbf{Soundness:}
If the statement is false, no cheating \(P\) can convince \(V\) except a small probability.
\end{itemize}
\subsubsection{Example: Quadratic Nonresidue Problem}
\label{sec:org597f587}

An integer \(a\) is a quadratic residue modulo \(n\) if there exists an integer \(x\) such that:
 \[
    x^2 \equiv a \pmod{n}
    \]
 An integer \(a\) is a quadratic nonresidue modulo \(n\) if no such integer \(x\) exists.
 Suppose A (Prover) claim that \(a\) is a \textbf{quadratic nonresidue}, and the B (Verifier) wants to check that using an \textbf{Interactive Proof System}.

\paragraph{An Interactive Proof System to the QAP can be:} B begins by choosing m random numbers \(\{r_1,r_2,...,r_m \}\). For each \(i\), \(1 \leq i \leq m\), he flips a coin:
\begin{itemize}
\item If it comes up heads he forms \(t = r_i^2 \pmod{m}\).
\item If it comes up tails he forms \(t = a \times ri^2 \pmod{m}\).
\end{itemize}

Then B sends \(t_1, t_2,..., t_m\) to A, who having unrestricted computing power, finds which of \(t_i\) are quadratic residues, and uses this information this information to tell B the results of his last m coin tosses. If this information is correct, B accepts.

\paragraph{Why this Will Work?}
\begin{itemize}
\item If \(a\) is really a quadratic nonresidue:

According to the property of quadratic nonresidue:
\begin{itemize}
\item \(t = a \times r_i^2 \pmod{m}\) is a quadratic non-residue.
\item \(t = r_i^2 \pmod{m}\) is a quadratic residue.
\end{itemize}
\(P\) can distinguish which side of the coin by looking whether t is a quadratic nonresidue or residue.

\item If \(a\) is a quadratic residue (\(P\) is lying):

Then both \(t = a \times r_i^2 \pmod{m}\) and \(x = r_i^2 \pmod{m}\) are quadratic residues.

Which means \(t_i\) are just random quadratic residues, all \(t_i\) looks the "same" for \(P\) to guess the coin side, \(P\) will respond correctly in the last part of the computation with probability \(1/2^m\).
\end{itemize}

\subsection{The Power of Interactive Proof Systems}
As mentioned in Section \ref{Limitation of the NP Proof-System}, While traditional  \(\mathsf{NP}\)  proof systems are powerful for verifying the existence of solutions, they encounter significant limitations when it comes to proving the non-existence of solutions. Specifically,  \(\mathsf{NP}\)  proof systems are inherently designed to provide short certificates for \textit{yes}-instances of decision problems. However, for \textit{no}-instances, such as those in the class \(\mathsf{coNP}\), no analogous short certificates are known. A quintessential example is the \textsc{Co-Hamiltonian Path} problem (\textsc{Co-HAMPATH}), where we seek to verify that a given graph does not contain a Hamiltonian path. Currently, it remains unknown whether \textsc{Co-HAMPATH} resides in  \(\mathsf{NP}\) , primarily because we lack efficient methods to certify the absence of a Hamiltonian path.

\begin{theorem}[\cite{10.1145/146585.146609}]
\(\mathsf{coNP} \subseteq \mathsf{IP}\)
\end{theorem}

This theorem, established by Shamir in 1990, reveals a profound capability of interactive proof systems: they can handle the complements of  \(\mathsf{NP}\)  problems efficiently. In other words, for every problem in \(\mathsf{coNP}\), there exists an interactive proof system where a \(P\) can convince a \(V\) of the truth of a statement without \(V\) needing to check an exhaustive list of possibilities. This inclusion signifies that interactive proof systems transcend the limitations of traditional  \(\mathsf{NP}\)  proof systems by enabling the verification of \textit{no}-instances through interactive protocols.

The detailed proof of this theorem leverages the \textit{sum-check protocol}, a pivotal technique in interactive proofs that facilitates the verification of complex statements through a series of interactive rounds between \(P\) and \(V\). For an in-depth exploration of the sum-check protocol and its role in proving that \(\mathsf{coNP} \subseteq \mathsf{IP}\), we refer the reader to Section~\ref{coNP-IP}.

\paragraph{In Practice:}
By assuming \(P\) has unlimited computing power, The theoretical model introduced in Interactive Proof Systems can describe many languages that cannot be captured using \(\mathsf{NP}\) model.
But back to practice, seems that this model will only work in theory until those kinds of unlimited computing machine comes in real life. (Can determine \(\mathsf{NP}\) problems in polynomial Time).
(i.e., Can quickly determine \(\mathsf{NP}\) problems like whether t is a quadratic nonresidue or residue)
Is that true?

\subsection{Secret Knowledge}
\label{sec:org811680d}
It is true that having unlimited computing machine is infeasible in current practice.
However, by assuming \(P\) runs in polynomial time with some "secret knowledge" that can help it communicating with \(V\) efficiently.
It can convince \(V\) that \(P\) has that "secret knowledge" without revealing it.
\subsubsection{Eyewitness \& Police Officer}
\label{sec:orgc4547a0}
Let us try to illustrate the above ideas using an informal example:
Assume that a crime \(x\) has happened, B is a police officer and A is the only eyewitness.
A is greedy in telling B that to tell him about what happened in \(x\), \$100,000 must be transferred to his bank account first.
For B, it is important to verify whether A has that "secret knowledge" -- details of crime \(x\) before making transfer.
And for obvious reasons, A cannot just prove that he has that "secret knowledge" by telling it directly to the police officer B.
By using interactive proof systems, A can convince B that he has that "secret knowledge" \(x\) without revealing it.
\subsubsection{Quadratic Nonresidue Problem}
\label{sec:orgebcccd9}
In the interactive proof system describe in Quadratic Nonresidue Problem, the key challenge for \(P\) (A) is to determine whether each number \(t_i \) sent by \(V\) (B)
is a quadratic residue or a quadratic nonresidue modulo m. This determination is crucial because it allows \(P\) to infer the results of B's coin tosses and respond correctly.
In this case, the "secret knowledge" for efficient computation on A is the \textbf{prime factorization of the modulus \(m\)}\footnote{Note: In practice, the modulus m is often chosen such that its factorization is hard to obtain (e.g., a product of two large primes),
ensuring that without the secret knowledge, determining quadratic residuosity remains difficult.}.
If \(P\) knows the prime factors of \(m\), they can efficiently compute the \emph{Legendre}\footnote{Legendre, A. M. (1798). Essai sur la théorie des nombres. Paris. p. 186.} or \emph{Jacobi}\footnote{Jacobi, C. G. J. (1837). "Über die Kreisteilung und ihre Anwendung auf die Zahlentheorie". Bericht Ak. Wiss: 127–136.} symbols to determine quadratic residuosity.
This "secret knowledge" enables polynomial \(P\) (A) to interact with B that are otherwise computationally infeasible.
\[
    (\text{"secret knowledge"} + \text{poly-time machine} \equiv \text{unlimited computing power} )
\]

\subsection{The Knowledge Complexity}
\label{sec:org1500878}
\subsubsection{The Knowledge Computable from a Communication}
\label{sec:org8ec6d03}
Which communications (interactions) convey knowledge?
\cite[Section 3.2]{DBLP:conf/stoc/GoldwasserMR85} Informally, those that transmit the output of an unfeasible computation. How to ensure that \(V\) gains no secret knowledge beyond the validity of the statement being proven?
\subsubsection{Knowledge Complexity}
\label{sec:org2120722}

\paragraph{Simulator:}

\cite{DBLP:conf/stoc/GoldwasserMR85} introduce the idea of \textbf{simulator}: An algorithm that can generate transcripts of the interaction without access to \(P\)'s secret information.

\paragraph{Quantifying Knowledge}

By linking the \textbf{chance} that \(V\) is able to distinguish the simulator, we quantify the knowledge complexity of proofs.
If the simulator can effectively replicate the interaction, making it indistinguishable to \(V\), \textbf{the knowledge complexity} is considered zero.

This formalization allows us to assess and prove the zero-knowledge property of certain interactive proofs by showing that sometimes \(V\) cannot gain knowledge because whatever it sees could be simulated without \(P\)'s help.

The Simulator acts as an algorithm that can generate transcripts of the interaction between \(P\) (A) and \(V\) (B) without \textbf{knowing \(P\)'s secret knowledge}.
By producing transcripts that are indistinguishable from those of a real interaction, the simulator demonstrates that \(V\) gains no additional knowledge from the interaction.
Therefore, if such a simulator exists, we say that the proof is zero-knowledge because any information \(V\) receives could have been simulated without \(P\)'s secret.

\subsubsection{Zero Knowledge Interactive Proof System for the QRP}
\label{sec:org4e566b5}
\cite[Section 4.2]{DBLP:conf/stoc/GoldwasserMR85} introduces a zero knowledge IP system for the quadratic residue problem by carefully designing the protocol and demonstrating the existence of a poly-time simulator. The difficulties of the proof is that M must compute the coin tosses correctly as a real \(P\) (A) with secret knowledge does. Since the simulator M simulates both sides of the interaction, it both can know/control the randomness of the coin.

\newpage

\section{The Power of Low-Degree Polynomials}
\label{sec:org9da4373}

This section of the survey builds upon the concepts presented by Justin Thaler \cite{thalerppt} during the Proofs, Consensus, and Decentralizing Society Boot Camp in 2019. Thaler's insightful discussion on the power of low-degree polynomials in verifiable computing serves as the backbone for the detailed explanations and proofs provided herein.

\subsection{Example: Equality Testing}\label{equality}
\label{sec:org6f7440e}
Two parties (i.e., Alice and Bob) each have an equal-length binary string:
\[
a = (a_1, a_2, ..., a_n) \in \{0, 1\}^n \mid  b = (b_1, b_2, ..., b_n) \in \{0, 1\}^n.
\]
They want to collaborate with each other (No malicious user) to determine whether \(a \equiv b\), while exchanging as few bits as possible.

\subsubsection{A trivial solution}

Alice sends a to Bob, who checks whether \(a \equiv b\). The communication cost is \(n\), which is optimal amongst deterministic protocols.

\subsubsection{A logarithmic cost randomized solution}

According to \cite[Section 2.3]{justin}, let \(F\) be any finite field with \(|F| \geq n^2\), then we interpret each \(a_i, b_i\) as elements of \(F\): Let \(p(x) = \sum_{i=1}^{n} a_i \times x^i\) and \(q(x) = \sum_{i=1}^{n} b_i \times x^i\)
\begin{enumerate}
\item Alice picks a random \(r\) in \(F\) and sends \((r, p(r))\) to Bob.
\item Bob calculates \(q(r)\), outputs EQUAL iff. \(p(r) \equiv q(r)\), otherwise output NOT-EQUAL.
\begin{itemize}
\item \textbf{Total Communication Cost: \(\mathcal{O}(\log n)\) bits}

Since there are at least total \(n^2\) elements in \(F\), to represent each of the elements, we need \(\log (|F|) = \log (n^2) = 2 \log (n) = \mathcal{O}(\log n)\) bits.
\end{itemize}
\item If \(a \equiv b\):
Then Bob outputs EQUAL with probability 1
\item If \(a \neq b\):
Then Bob outputs NOT-EQUAL with probability at least \((1 - \frac{1}{n})\) over the choice of \(r\) in \(F\).
\end{enumerate}

A detailed proof of this statement will be given in the section \ref{low-degree-polynomials}.

\subsection{Low-degree Polynomials}\label{low-degree-polynomials}
\label{sec:org504c80a}

\subsubsection{Field}
Field arises from the need for a structured and versatile system in mathematics and science to perform algebraic operations in a consistent and predictable way.

A field is a set equipped with two operations, addition and multiplication, along with their respective inverses, subtraction and division (except division by zero). Operations in a field follow specific rules, such as commutativity, associativity, and distributivity, ensuring that the result of any operation between elements of the field remains within the field itself (closure). Here are some non-field examples:
\begin{itemize}
\item The set of 2x3 matrices cannot perform multiplication.
\item The set of 2x2 matrices, multiplication between any two elements is not commutative.
\item Some of the elements in set Z/6Z (integers modulo 6) do not have multiplicative inverses.

A number a in Z/6Z has a multiplicative inverse if there exists a number b in Z/6Z such that \(a \times b \equiv 1 \pmod 6\).
For example, 2, 3, 4 do not have a multiplicative inverse, in fact, integers mod p (Z/pZ) is a field when p is a prime number.
\end{itemize}

\subsubsection{Reed-Solomon Error Correction}\label{reed}

Since there are total \(n\) bits of \(a\) and \(b\), the lowest-degree polynomials that for us can ensure the uniqueness of representation of each \(a_i\) and \(b_i\) is \(n\).
Which means each bit \(a_i\) (and correspondingly \(b_i\)) is uniquely represented as the coefficient of a distinct term in the polynomial and affects the polynomial differently. And that is why we need to define \(p(x) \) and \(q(x) \) in the following way for error detection in equality testing:

\[
p(x) = \sum_{i=1}^{n} a_i \times x^i
\]
\[
q(x) = \sum_{i=1}^{n} b_i \times x^i 
\]

\begin{theorem}
    Any non-zero polynomial \(d(x)\) of degree \(n\) has at most \(n\) roots
\end{theorem}

\begin{proof}
    Assume the polynomial has more than \(n\) roots. Let \(r_1, r_2, ..., r_{n+1}\) be distinct elements of the field \(F\), such that \(d(r_i) = 0\) for each \(i = 1, 2, ..., n+1\). 
Then the polynomial d(x) can be written as:
\[
   d(x) = (x - r_1)(x - r_2) ... (x - r_{n+1})q(x) 
   \]
where \(q(x)\) is some polynomial of degree \(m \geq 0\) and \((x - r_i)\) are factors corresponding to the roots. Then the product of \((x - r_1)(x - r_2) ... (x - r_{n+1})\) is a polynomial of degree \(n + 1\). This assumption leads to a contradiction. Hence the polynomial \(d(x)\) can have at most \(n\) distinct roots.
\end{proof}

\begin{theorem}
    In section \ref{equality}, if \(a \neq b\), then the probability of Bob is wrong is at most \(\frac{1}{n}\)
\end{theorem}

\begin{proof}
In equality testing, if \(p(x) \neq q(x) \), then the chance that Alice picks a random \(r \) in \( F \) such that \(d(r) = p(r) - q(r) \equiv 0 \) is at most \(( \frac{n}{|F|} \leq \frac{n}{n^2} \leq \frac{1}{n})\). The reason is that a n-degree polynomial \(d(r) = p(r) - q(r)\) has at most \(n\) roots, a randomly picked value \(r\) in \(F\) of size \(n^2\) will only let \(d(r)\) equal to zero with a probability \(( \frac{n}{n^2} = \frac{1}{n} )\).
\end{proof}

\subsubsection{Polynomials are Constrained by Their Degree}

As we can see that, A polynomial \(d(x)\) of degree \(n\) over a large field \(F\) is uniquely determined by its values on \(n+1\) distinct points. (one extra constant coefficient with no variable)

If \(p(x)\) is not equal to \(q(x)\), then they can only agree on at most n points (\(d(x) = p(x) - q(x) = 0 )\), meaning they \textbf{differ at most everywhere} on the field. (Any two polynomials of degree \(n\) can agree in at most \(n\) places, unless they agree everywhere.)
This strong divergence is very useful and powerful for error detection.

\paragraph{The Power of Low-degree Polynomials}
In practice, we aim to keep the degree of a polynomial as low as possible while ensuring that each term uniquely affects the polynomial.
A polynomial of degree \(n\) has \(n+1\) terms, each with a distinct power of \(x\) and a unique coefficient, which ensures that every term influences the polynomial differently. Low-degree polynomials are powerful because if two polynomials differ, they will diverge over many points when evaluated multiple times with efficient evaluation.
This makes discrepancies clear across a larger number of evaluations. In cryptography, this property allows for efficient detection of errors or differences, especially when random evaluations are used.

\subsection{Example: Freivalds's algorithm for Verifying Matrix Products}
\label{sec:org64df47b}
Input are two \(( n \times n )\) matrices A, B. The goal is to verify the correctness of \(A \cdot B\). 
The time complexity of matrix multiplication is \(\mathcal{O}(n^3)\), this is because each element in the resulting matrix \(A \cdot B\) is computed by taking the dot product of a row of \(A\) and a column of \(B\).
For each of \(n^2\) elements in the resulting matrix, you perform n multiplications and additions, leading to a total of \(\mathcal{O}(n^3)\) operations.
The best bound of matrix multiplication algorithm for now is \(\mathcal{O}(n^{2.371552})\) \cite{vir}.

\noindent If a \(P\) claims the answer of \(A \cdot B \) is a matrix \(C\)? Can V verify it in \(\mathcal{O}(n^2)\) time?

\paragraph{\(\mathcal{O}(n^2)\) Protocol:\cite{freivalds}}
\begin{enumerate}
\item V picks a random \(r\) in \(F\) and lets \(x = (r, r^2, ..., r^n)\).
\item V computes \(C \cdot x\) and \(A \cdot ( B \cdot x )\), accepting if and only if they are equal.
\end{enumerate}

\paragraph{Runtime Analysis:}

V's runtime dominated by computing 3 matrix-vector products, each of which takes O(n\textsuperscript{2}) time.
\begin{itemize}
\item \(C \cdot x\) is one matrix \((n \times n)\) times a vector \((n \times 1)\), the time complexity is \(\mathcal{O}(n^2)\).

Because each row of B is multiplied by the vector \(x\), requiring \(n\) multiplications and \(n\) additions per row, and there are \(n\) rows.

\item \((A \cdot B) \cdot x = A \cdot (B \cdot x)\) takes two matrix-vector multiplications.

Matrix multiplication is associative, \(B \cdot x\) takes \(\mathcal{O}(n^2)\) first, and produces a \((n \times 1)\) matrix \(M\), then \(A \cdot M\) will also take \(\mathcal{O}(n^2)\).
\end{itemize}

\paragraph{Correctness Analysis:}
\begin{itemize}
\item If \(C \equiv A \cdot B\):

Then V accepts with probability 1
\item If \(C \neq A \cdot B\):

The V rejects with high probability at least \((1 - \frac{1}{n})\).
\end{itemize}

\begin{proof}[Simplified Proof]
Recall that \(x = (r, r^2, ..., r^n)\). So each matrix-vector multiplication is the polynomials we've seen in section \ref{reed} Reed-Solomon Error Correction at \(r\) of the i-th row of C. So if one row of \(C\) does not equal the corresponding row of \(A \cdot B\), the fingerprints for that row will differ with probability at least \(( 1 - \frac{1}{n} )\), causing V to reject w.h.p.    
\end{proof}

\subsection{Function Extensions}
\label{sec:org3368b97}

\subsubsection{Schwartz-Zippel Lemma}
\label{sec:orge739764}

The Schwartz-Zippel Lemma\cite{sch} is an extension of univariate error-detection to multivariate polynomials.
If \(p\) and \(q\) are distinct l-variate polynomials of \textbf{total degree} at most \(d\). Then the same kind of statement holds.
If we evaluate them at randomly chosen inputs, they agree at the probability at most \(\frac{d}{|F|}\).

\paragraph{Total Degree:} The total degree of a polynomial is the maximal of the sums of all the powers of the variables in one single monomial.

\noindent For example: \(\text{deg}(x^2yz^4 - 3y + 4xe^5 - xy^3z^2) = 7\) (first monomial).

\subsubsection{Extensions}
An extension polynomial bridges the gap between a function defined on a discrete set of points and a function defined over a continuous (or larger discrete) domain.

A l-variate polynomial \(g\) over \(F\) is said to extend \(f\) if and only if \(g\) agrees at all of the input where \(f\) is defined. For example: We are given a function \(f\) that maps l-bits binary strings to a field \(F\). This means \(f\) is defined on all possible combinations of \(l\) bits (\(\{0, 1\}^l\)).
But it is only defined on a finite set of points (the binary strings), which usually cannot form a field, \textbf{where we can leverage algegratic tools of polynomials}.

\paragraph{A function \(g\) is said to extend \(f\) if:}

\begin{itemize}
    \item For all \( x \in \{0, 1\}^l \), \( f(x) = g(x) \).
    \item \( g \) is defined on a larger field
\end{itemize}

Let's say \(l = 1\), then \(f\) is defined on input set \(\{0, 1\}\), where \(f(0) = 2\), \(f(1) = 3\). If we want to extend \(f\) to field \(R\) (real numbers), then our objective is to find a polynomial \(g(x)\) in R such that \(g\) agrees with \(f\) on all inputs where \(f\) is initially defined. (i.e., \(g(0) = 2\), and \(g(1) = 3\)).
We can find \(g(x) = x + 2\) defined for all \(x\) in \(R\), not just \(\{0, 1\}\).

By representing \(f\) as a polynomial \(g\), \textbf{we can apply the rich toolbox of algebraic methods and theorems available for polynomials}. For example, Schwartz-Zippel Lemma is more powerful when there is a low-degree extension represents that function.

\subsubsection{Constructing Low-Degree Extensions}\label{lde}
In this section, we present a general way to construct low-degree extensions.

\paragraph{Notations:} 
There is a vector \(w = (w_0, w_1,..., w_{k-1})\) in \(F^k\). \(W: H^m \to F\):
We define a function \(W: H^m \to F\) such that \(W(z) = w_{a(z)}\) if \(a(z) \leq k-1\), and \(W(z) = 0\) otherwise.

\begin{itemize}
    \item \(W\) acts as a way to represent the vector \(w\) as a function over \(H^m\) that for indices corresponding to elements of \(w\), \(W(z)\) returns the corresponding \(w_i\), otherwise \(0\).
    \item \(a(z): H^m \to F\) is the lexicographic order of \(z\), which means \textbf{transforming an m-element vector to an index in the original \(w\)}.
\end{itemize}

\paragraph{Low-Degree Extension \(\widetilde{W}: F^m \to F\):} \(\widetilde{W}\) is an extension of \(W\) input from \(H^m\) to \(F^m\), such that \(\widetilde{W}\) is a polynomial of degree \textbf{at most \(|H|-1\)} in each variable, which enables efficient computation and has nice algebraic properties.

\noindent \textbf{The degree of \(W\) is at most |H|-1 can be understood in the following cases:}

\paragraph{Univariate Case: \(W(x): H \to F\): }
\begin{itemize}
    \item We have \(n = |H|\) distinct points in the subset H.
    \item For each \(h \in |H|\), we want \(\widetilde{W}(x): F \to F)\) to satisfy \(\widetilde{W}(h) \equiv W(h)\).
\textbf{A univariate polynomial of degree at most \((|H|-1)\) can be uniquely determined by its values on those \(n\) points}.
\end{itemize}

\paragraph{Multivariate Case: \(W(x): H^m \to F^m\)}
\begin{itemize}
\item Similarly, we have a m-dimensional grid \(H^m\), where \(H \subseteq F\) and \(|H| = n\).
\item A polynomial \(\widetilde{W}(x_1, x_2,..., x_m)\) (in m variables) that:
\begin{enumerate}
\item Agrees with \(W\) on every point in \(H^m\) (i.e., \(\widetilde{W}(h) = W(h)\) for all \(h \in H^m\).
\item Has degree at most \(|H| - 1\) in each variable.
(i.e., for each variable \(x_i\), the highest exponent of \(x_i\) in \(\widetilde{W}\) is \(\leq |H|-1\).
In other words, just like in the univariate case (where we need \(deg(\widetilde{W}) \leq |H|-1\) to interpolate \(|H|\) points), here each variable is similarly bounded by \(|H|-1\). This ensures \(\widetilde{W}\) can uniquely "pass through" all the points specified by \(W\) on \(H^m\).
\end{enumerate}
\end{itemize}

\noindent The low-degree extension is the \textbf{simplest} polynomial that fits all the given points in \(H^m\), The size of \(H\) determines how "complex" the polynomial needs to be (i.e., degree) in order to pass through all those points without ambiguity.

\paragraph{Here is the full definition of \(\widetilde{W}: F^m \to F\):}

\[
\widetilde{W}(t_1, ..., t_m) = \sum_{i=0}^{k-1} \widetilde{B_i}(t_1, ..., t_m) \cdot w_i.
\]

\paragraph{Indicator Function: \(\widetilde{B_i}: F^m \to F\)}

The polynomials \(\widetilde{B_i}\) act as an indicator functions on \(H^m\), on that field, \(\widetilde{B_i}(z) = 1\) if and only if \(i \equiv a(z) = a(t_1,...,t_m)\). Otherwise \(\widetilde{B_i}(z) = 0\).

Outside \(H^m\) (in \(F^m / H^m\), \(\widetilde{B_i}\) takes on values determined by its polynomial extension, which means when input is outside \(H^m\), \(\widetilde{W}\) can have sum of multiple terms.

\paragraph{Sum Selection: \(\widetilde{W}: F^m \to F\)}

\(\widetilde{W}\) is a low-degree polynomial, summing all possible \(k\) indexes \(i\) from \(0\) to \(k-1\), in the field of \(H^m\).
According to the definition of \(\widetilde{B_i}\) there will be only one "selected" corresponding \(w_i\) which is equal to \(\widetilde{W}\).

\paragraph{Lagrange Basis Polynomial: \(\widetilde{B}(z, p): F^m \to F\)}\label{lagrange}

Also, we can express \(\widetilde{W}(t_1, ..., t_m)\) as:

\[
\widetilde{W}(z) = \sum_{p \in H^m} \widetilde{B}(z, p) \cdot W(p)
\]
  Which constructs the polynomial \(\widetilde{W}\) by summing the contributions from all basis polynomials \(\widetilde{B}(z, p)\), each weighted by \(W(p)\).
\begin{itemize}
\item \(\widetilde{B}(z, z) = 1\), \(\widetilde{B}(z, p) = 0\) for all \(z \in H^m\) where \(z \neq p\)
\item When \(z\) outside \(H^m\), \(\widetilde{B}(z, p)\) can be other values
\end{itemize}

This means, when \(z\) is in \(F^m / H^m\), each \(\widetilde{B}(z, p)\) is a polynomial in \(z\) and can be evaluated at any \(z\) in \(F^m\). And To compute \(\widetilde{W}(z)\), which means by summing over all \(W(p) \in H^m\) that has valid \(\widetilde{B}(z, p)\).
Since \(W\) is only defined in the field \(H^m\). And the low-degree polynomial property still holds.
By enlarging the input field \((F^m)\) while keeping the degree of \(\widetilde{W}\) low, this way of constructing LDE \(\widetilde{W}\) helps us effectively using the power of randomness to detect the potential error.

\subsubsection{Multilinear Extensions}

\label{sec:org9c8057d}
A multilinear extension of a function \(f: \{0, 1\}^n \to F\) (where \(F\) is a finite field) is a polynomial \(\bar{f}: F^n -> F\) that agrees with \(f\) on \(\{0, 1\}^n\) and is \textbf{multilinear}, which means each variable \(x_i\) in \(\bar{f}\) has a degree at most 1, which make them highly effective to evaluate. (a univariate \(f\) by make other variables constants will become a linear function).

And since multilinear polynomials have minimal degree, the error detection probability is maximized for a given field size.

\section{The Sum-Check Protocol}
\label{sec:org46e8bf2}
Suppose given a \(l\)-variate polynomial \(g\) defined over a finite field \(F\). The purpose of the \textbf{sum-check protocol} \cite{89518} is to compute the sum:
\[
  H := \sum_{b_1 \in \{0, 1\}} \sum_{b_2 \in \{0, 1\}} ... \sum_{b_l \in \{0, 1\}} g(b_1, b_2,...,b_l)
\]

In applications, this sum will often be over a large number of terms, so \(V\) may not have the resources to compute the sum without help.
Instead, she uses the \textbf{sum-check protocol} to force \(P\) to compute the sum for her. \(V\) wants to verify that the sum is correctly computed by \(P\), where \(g\) is a known multivariate polynomial over finite field F.

\subsection{Initialization}
\label{sec:org69b321a}
P claims that the total sum equals a specific value \(H_0\).

\subsection{First Round}
\label{sec:orgd7df33b}

\begin{algorithm}
\captionsetup{labelfont=bf, name=Protocol}
\caption{First Round of the Sum-check Protocol}
\begin{enumerate}
\item P sends a univariate polynomial \(s_1(x_1)\) to V, which is claimed to equal:
\label{sec:org69d9018}
\[
s_1(x_1) := \sum_{b_2 \in \{0, 1\}} \sum_{b_3 \in \{0, 1\}} ... \sum_{b_l \in \{0, 1\}} g(x_1, b_2,...,b_l)
\]

\item V calculates \(s_1(0) + s_1(1)\) and checks whether that value is equal to \(H_0\).
\label{sec:org214ae12}

Since \(s_1\) is a univariate polynomial, V can compute \(s_1(0) + s_1(1)\) directly (not using structure of \(H_0\)) in relatively short amount of time.

\item V picks a random element \(r_1\) from F and sends it to P.
\label{sec:org564094a}
\item V sets \(H_1 := s_1(r_1)\) for use in the next iteration.
\label{sec:orgd833d82}
\[
H_1 = s_1(r_1) := \sum_{b_2 \in \{0, 1\}} \sum_{b_3 \in \{0, 1\}} ... \sum_{b_l \in \{0, 1\}} g(r_1, b_2,...,b_l)
\]
\end{enumerate}

\end{algorithm}

\newpage

\subsection{Iterative Rounds \(( i = 2\) to \(l )\)}
\label{sec:org6154ac8}
For each round i (i from 2 to l):
\begin{algorithm}
\captionsetup{labelfont=bf, name=Protocol}
\caption{Round i to l of the Sum-check Protocol}
\begin{enumerate}
\item P sends a univariate polynomial \(s_i(x_i)\) to V, claimed to equal:
\label{sec:orgffbb8b2}
\[
s_i(x_i) := \sum_{b_{i+1} \in \{0, 1\}} ... \sum_{b_l \in \{0, 1\}} g(r_1,...,r_{i-1}, x_i, b_{i+1},...,b_l)
\]
  Which represents the partial sum over variables \(b_{i+1}\) to \(b_l\), with \(b_1\) to \(b_{i-1}\) fixed to random values chosen by \(V\) in previous rounds.

\item V calculates \(s_i(0) + s_i(1)\) and checks whether that value is equal to \(H_{i-1}\).
\label{sec:org27dbdd6}
\(H_{i-1}\) is the sum in the previous iteration: \(s_{i-1}(r_{i-1})\):

\[
H_{i-1} := s_{i-1}(r_{i-1}) := \sum_{b_{i} \in \{0, 1\}} ... \sum_{b_l \in \{0, 1\}} g(r_1,...,r_{i-1},b_i,...,b_l)
\]
\item V picks a random element \(r_i\) from \(F\) and sends it to P.
\label{sec:org4cd06e4}
\item V sets \(H_i = s_i(r_i)\) for use in the next iteration.
\label{sec:orgd8091e4}
\[
H_i = s_i(r_i) := \sum_{b_{i+1} \in \{0, 1\}} ... \sum_{b_l \in \{0, 1\}} g(r_1,...,r_i,b_{i+1},...,b_l)
\]
\end{enumerate}
\end{algorithm}

\subsection{Final Check}\label{final-check}
\label{sec:org87ac174}
In the final iteration: 
\[
H_l = s_l(r_l) := g(r_1, r_2, ..., r_l)
\]
All \(b_i\) in the original equation of \(g(b_1,...,b_l)\) has been fixed to the random number \(r_i\) chosen by V in previous \(l\) rounds.
V then checks whether \(s_l(r_l)\) is equal to \(g(r_1,r_2,...,r_l)\) \textbf{by calculating it} herself. \textbf{If \(s_l(r_l)\) is equal to \(g(r_1,...,r_l)\), V accepts}.

\subsection{Soundness of the Sum-Check Protocol}
\label{sec:org7e01221}

\begin{theorem}
    The Sum-check Protocol is Sound and Complete.
\end{theorem}

\noindent \textbf{Completeness} holds by design: If P sends the prescribed messages, then all of V's check will pass. Let's prove the \textbf{soundness} of the protocol:

\begin{proof}
The Sum-Check Protocol is \textbf{Sound}:

\noindent The Section \ref{final-check} Final Check is the most crucial part, and plays a key role in understanding this protocol.
In the last step, V directly computes \(g(r_1,r_2,...,r_l)\) and \(s_l(r_l)\) to check whether they are equal, note that this is the only time for V to actually compute \(l\)-variate polynomial \(g\) by herself.

This direct comparison is critical because it anchors the entire verification process to the actual function \(g\).
According to \hyperref[sec:orge739764]{Schwartz-Zippel Lemma}, in the final check, if \(s_l \neq g\), then the probability \(g(r_1,...,r_l)\) equal to \(s_l(r_l)\) is less than \(d/|F|\). (\(d\) is the Total Degree of polynomial \(g\) and \(s_l\)).

\paragraph{Backward Reasoning:}
\label{sec:orgbc212aa}

By working backwards from the last iteration, We can understand the correctness of each step based on the validity of the final check:

In the last step, if the equality of \(g(r_1,...,r_l)\) and \(s_l(r_l)\) holds, \textbf{with high probability, \(s_l(x_l)\) must be the correct polynomial of \(g(r_1,...,x_l)\)},
because a dishonest \(P\) would need to guess \(r_l\) to fake \(s_l(x_l)\) in order to pass the test, and the chance of success is less than \(d/|F|\), where \(d\) is the totdal degree of polynomial \(g\) over field \(F\).

\noindent If V can confirm \(s_l(x_l)\) is correctly formed w.h.p in \hyperref[sec:org7228d30]{Final Check} where:
\[
s_l(x_l) := g(r_1,...,r_{l-1},x_l)
\]
\noindent Then in iteration \((l  - 1)\), \(s_{l-1}(r_{l-1})\) can be written as:
\begin{eqnarray*}
s_{l-1}(r_{l-1}) &=& \sum_{b_l \in \{0, 1\}} g(r_1,...,r_{l-2}, r_{l-1}, b_l) \\
                 &=& g(r_1,...,r_{l-1},0) + g(r_1,...,r_{l-1},1) \\
		 &=& s_l(1) + s_l(2)
\end{eqnarray*}

\noindent Thus, \(s_{l-1}(x_{l-1})\) must be correctly formed w.h.p based on correct \(s_l(x_l)\).

\paragraph{Inductive Case: }

If \(s_i(x_i)\) is correctly formed in round \(i\)
\label{sec:orgfc85d55}
For any round \(i\), if we can make sure \(s_i(x_i)\) is correctly formed w.h.p that:
\[
s_i(x_i) = \sum_{b_{i+1} \in \{0, 1\}}...\sum_{b_l \in \{0, 1\}} g(r_1,...,r_{i-1}, x_i, b_{i+1},...,b_l)
\]
Backwards to its previous round:
\begin{eqnarray*}
s_{i-1}(r_{i-1}) &=& \sum_{b_i \in \{0, 1\}} ... \sum_{b_l \in \{0, 1\}} g(r_1,...,r_{i-1},b_i,...,b_l) \\
                 &=& \sum_{b_{i+1} \in \{0, 1\}} ... \sum_{b_l \in \{0, 1\}} g(r_1,...,r_{i-1},0,b_{i+1},...,b_l) \\
                 &+& \sum_{b_{i+1} \in \{0, 1\}} ... \sum_{b_l \in \{0, 1\}} g(r_1,...,r_{i-1},1,b_{i+1},...,b_l) \\
		 &=& s_i(1) + s_i(2)
\end{eqnarray*}
Then, with high probability, \(s_{i-1}(r_{i-1})\) should be correct, which indicate \(s_{i-1}(x_{i-1})\) should also be correctly formed.
\textbf{By induction, in the first round, since \(s_2(x_2)\) is correctly formed, then w.h.p \(s_1(x_1) = s_2(0) + s_2(1)\) should be formed correctly.
And thus w.h.p, \(H_0 = s_1(0) + s_1(1)\) should be the correct answer.}
So far, we've proved the soundness of the protocol.
\end{proof}

\paragraph{Summary:}
We've proved the soundness of the sum-check protocol by showing that:
According to Schwartz-Zippel Lemma, with high probability a dishonest P cannot initially fake a incorrect \(H_0\)
that won't break the consistent in every round of the protocol and causes a correct \(s_l(r_l)\) equals to \(g(r_1, ..., r_l)\) in the final round.

\newpage

\subsection{Analyzing the Sum-Check Protocol}
\label{sec:org5fcb437}

\subsubsection{A Scenario where V cannot Compute \(g(r_1,...,rl)\)}
The soundness is relying on the final check, by validating the final step, V effectively validates all prior steps due to their interdependence in the chain of validations.
For example, here we show a specific cheating strategy when V cannot compute \(g(r_1,...,r_l)\) by herself:

Imagine a dishonest P consistently uses different function \(j(x)\) instead of the correct function \(g(x)\) throughout the protocol.
In this scenario, P computes all the partial sums and polynomials correctly with respect to \(j(x)\), ensuring consistency at each step, P only deviates from the correct computation at the final check when V computes \(g(r_1,...,r_l)\).

During the protocol, V only computes \(g(x)\) in the final check, the only point of detection is the final check, and the probability of detection is \(\frac{d}{|F|}\).
But if there is no final check, which means the protocol ends at round \(l\), \textbf{V cannot have any guess of what function P used to compute \(H_0\)}.

\subsubsection{Probability of Successful Cheating by P in Sum-check Protocol}

In this part, we are going to quantify the probability for a dishonest P can successfully convince V for the wrong computation \(H_0\).
\begin{claim}
For sum-check protocol on polynomial \(g: F^l \to F\) with total degree \(d\), if P is cheating at any round, the upper bound for V to accept is: \(ld/|F|\).
\end{claim}

\begin{proof} For any round \(i\), let's assume what a dishonest P sends to V (\(s_i(x_i)\)) is formed incorrectly (i.e., \(s_i(0) + s_i(1) \neq s_{i-1}(r_{i-1})\), deviation at round i), then.

In round \(i\), \(s_i(r_i)\) is set to \(H_i\) by V (in round \(l\), \(H_i\) becomes to \(g(r_1,...,r_l)\)).
\textbf{The probability for P to provide \(s_i\) to satisfy polynomial \(H_i - s_i(r_i) = 0\) with randomly selected \(r_i\) by V from field \(F\) is \(\frac{d}{|F|}\)}, where \(d\) is the total degree of polynomial \(g\) and \(s\), according to \hyperref[sec:orge739764]{Schwartz-Zippel Lemma}.

\textbf{If \(s_i(r_i) \neq H_i\), P is left to prove a false claim in the recursive call:} P must construct \(s_{i+1}(x_{i+1})\) such that \(s_{i+1}(0) + s_{i+1}(1) = s_i(r_i)\). This means \(s_{i+1}(r_{i+1})\) must deviate from true \(H_{i+1}\), leading to a new error polynomial in subsequent rounds.

Thus, we can get the the cumulative probability of acceptance when P deviates at round \(i\) and V does not detect in rounds \(i+1\) to \(l\) and end up accepting: \(\frac{(l-i)d}{|F|}\).

This summation approach accounts for the fact that each round provides an independent opportunity for \(V\) to catch cheating, effectively adding an "\textbf{error margin}" with each additional round. Thus, more rounds increase the total soundness error linearly, rather than multiplicatively, then we have:
\begin{eqnarray*}
P[A] &\leq& P[D_i] + P[D_{i+1,l}] \\
     &\leq& \frac{d}{|F|} + \frac{(l - i)d}{|F|} \\
     &=& \frac{(l - i + 1)d}{|F|}
\end{eqnarray*}
Where \( A \) denote "V accepts", \( D_i \) denote "V does not detect at \( i \)", and \( D_{i+1,l} \) denote "V does not detect in \(\{i+1, \ldots, l\}\)".

\textbf{The possibility of acceptance when P deviates at round \(i\) is \((l-i+1)d/|F|\).}

\noindent The worse-case scenario is that if P deviates in the first iteration \((i=1)\) (Error Margin), the total probability of acceptance is:
\[
\frac{d}{|F|} + \frac{(l-1)d}{|F|} = \frac{ld}{|F|}
\]
\end{proof}

\subsubsection{Example: \(l = 2\)}
\label{sec:org4741d56}
Let's consider a minimum sum-check protocol with \(l = 2\) by working backward from the last iteration to understand the soundness:
\begin{enumerate}
\item Final Check:
\label{sec:org6464e57}
\begin{enumerate}
\item V computes \(g(r_1,r_2)\)
\item Comparison with P's \(s_2(r_2)\)

If the equality holds, with high probability, \(s_2(x_2)\) must be correct polynomial \(g(r_1,x_2)\), because a dishonest \(P\) would need to guess \(r_2\) to fake \(s_2(x_2)\).
\end{enumerate}

\item Round 2:
\label{sec:org998da04}
Since \(s_2(x_2)\) that P sends to V is confirmed to be correct, the sum: \(s_2(0) + s_2(1) = g(r_1,0) + g(r_1,1) = H_1\) must be satisfied.
This confirms that \(H_1\) is correctly computed based on \(s_2(x_2)\).
\item Round 1:
\label{sec:orga4f419e}
At the end of this round, P sets \(H_1 = s_1(r_1)\). Since \(H_1\) is now confirmed to be \(g(r_1,0) + g(r_1,1)\), it implies \(s_1(r_1) = g(r_1,0) + g(r_1,1)\),
thus \(s_1(x_1)\) is correct polynomial \(\sum_{b_2 \in \{0, 1}\} g(x_1, b_2)\) w.h.p, any deviation would be detected with high probability due to the random \(r_1\).
\item Initialization:
\label{sec:orgc7a256c}
V check \(s_1(0) + s_1(1) = H_0\)
Since \(s_1(x_1)\) is correct w.h.p, then \(H_0\) should be correctly formed based on \(s_1(x_1)\).

\end{enumerate}

\noindent \(V\) only needs to perform a few evaluations of univariate-polynomials and checks (except the final check \(g\)), making the protocol practical even for large computations.

\subsection{coNP \(\in\) IP}\label{coNP-IP}

Having introduced the sum-check protocol, we now demonstrate how it empowers interactive proof systems to verify statements in \(\mathsf{coNP}\) efficiently. Recall from Section~\ref{cook-levin} the \textit{Cook--Levin Theorem}, which establishes that the \textsc{SAT} problem is  \(\mathsf{NP}\) -complete. Consequently, any problem in  \(\mathsf{NP}\)  can be transformed into an instance of \textsc{SAT} in polynomial time. A canonical example is the \textsc{HAMPATH} problem, which asks whether a given directed graph \(G\) contains a Hamiltonian path.

In this subsection, we turn our attention to \emph{complements} of  \(\mathsf{NP}\)  problems, such as \textsc{Co-Hamiltonian Path} (\textsc{Co-HAMPATH}), which asks whether a directed graph \emph{does not} contain a Hamiltonian path. Traditional  \(\mathsf{NP}\)  proof systems do not straightforwardly provide a ``short certificate'' for the non-existence of certain structures (e.g., no Hamiltonian path). Therefore, their capability to verify statements in \(\mathsf{coNP}\) is limited.

By leveraging the sum-check protocol, we can  Specifically, we reduce any \(\mathsf{coNP}\) problem, such as \textsc{Co-HAMPATH}, to a corresponding \(\#\text{SAT}\) instance. For instance, proving that \(\#\text{SAT}(\phi) = 0\) (i.e., \(\phi\) has no satisfying assignments) corresponds to showing there is no Hamiltonian path in \(G\). The key insight is that verifying a \(\#\text{SAT}\) claim via the sum-check protocol can be done in \emph{polynomial time} with the help of an interactive \(P\)--\(V\) framework.

Therefore, by composing these steps---polynomial-time reduction to \(\#\text{SAT}\) plus an efficient interactive verification protocol (sum-check)---we conclude that every problem in \(\mathsf{coNP}\) admits an interactive proof system.

\begin{theorem}[\cite{10.1145/146585.146609}]
\(\#\text{SAT} \subseteq \mathsf{IP}.\)
\end{theorem}

\noindent Let \(\#\text{SAT}\) be the number of satisfying assignments of Boolean formula \(\phi\).
\[\#\text{SAT} = \{ \langle\phi, k\rangle \mid k = \#\phi\}\]

\begin{proof}
    \(\#\text{SAT} \subseteq \mathsf{IP}\) \cite{DBLP:books/daglib/0086373}[Chapter 10].
    
\noindent \cite{sip}Assume \(\phi\) has \(m\) variables \((x_1,...,x_m)\), let \(\phi(a_1,...a_i)\) be \(\phi\) with \(x_1 = a_1,\) ... \(,x_i = a_i\) for \(a_1,...,a_i \in \{0, 1\}\). Then we call \(a_1,...,a_i\) presets, the remaining \(x_{i+1},...,x_m\) stay as unset variables. Then equivalently:
\[
\#\phi(a_1,...,a_i) = \sum_{a_{i+1},...,a_{m} \in \{0, 1\}}{\phi(a_1,...,a_m)}
\]
\[
\#\phi(a_1,...,a_i) = \#\phi(a_1,...,a_i,0) + \#\phi(a_1,...,a_i,1)
\]
We can apply the sum-check protocol on \(\#\phi\) if \(\phi\) is a low-degree polynomial. 
By going gate-by-gate through \(\phi\), we can replace each gate with the gate's arithmetization:
\[
\text{NOT}(x) \to 1 - x
\]
\[
\text{AND}(x, y) \to x \times y
\]
\[
\text{OR}(x, y) \to x + y - x \times y
\]
To complete the proof of the theorem, we need only show that V operates in probabilistic polynomial time, let's use \(g(x_1,...,x_m): \{0, 1\}^m \to \{0, 1\}\) to denote the low-degree polynomial of \(\#\phi\). For \(0 \leq i \leq m\) and for \(a_1,...,a_i \in \{0, 1\}\), let
\[
g(a_1,...,a_i) = \sum_{a_{i+1},...,a_m \in \{0, 1\}}{\phi(a_1,...,a_m)}
\]
For Boolean formula \(\phi\) with S gates, where each gate can affect the overall degree at most 2 (\(\text{AND}(x, y) = x \times y) \). In the worst-case scenario (where there are all \(\text{AND}\) gates in \(\phi\)), the boolean formula \(\phi\), which is tree-like with each gate having a single output path and no shared sub-expressions, the degree of the polynomial representing \(\phi\) can be up to \(2^D\), where \(D\) is the depth of the formula. For a balanced formula, the depth \(D\) is in \(\mathcal{O}(\log{S})\), therefore, the degree of \(g\) is at most
\(2^{\log{S}} = \mathcal{O}(S)\). So for \(V\), in each round:

\paragraph{Communication Cost:} Since polynomial \(g: \{0, 1\}^m \to \{0, 1\}\) can be represented by its coefficients, requiring at most \(\mathcal{O}(S)\) elements (the degree of \(g\)), P sends a polynomial \(s\) of degree equal to \(g\) to V, which is in size \(\mathcal{O}(S)\), and V sends one field element \(r\) in each round. The total communication cost is \(\mathcal{O}(m \cdot S)\).

\paragraph{V Time Complexity:} 
It takes \(\mathcal{O}(S)\) time for V to process each of \(m\) messages (\(s\) in degree \(\mathcal{O}(S)\)) of P, and \(\mathcal{O}(S)\) time to evaluate \(g(r)\). The total cost is \(\mathcal{O}(m \cdot S)\).
    
\end{proof}

\noindent In other words, the non-existence of certain objects (e.g., a Hamiltonian path) can be certified interactively without requiring an exhaustive search over all possibilities by \(V\) using the interactive sum-check protocol.

\subsection{Limitation of the Sum-Check Protocol}\label{limitation}
However, while V's runtimes scales linearly with the circuit size, the practical limitation lies in the\textbf{ P's runtime}. In each of the m rounds, P must compute the univariate polynomial \(g_i\) by summing over all \(2^{m-i}\) possible assignments of the remaining \(m-i\) variables, since \(g\) has degree \(S\), P's runtime complexity becomes \(\mathcal{O}(Sm\cdot2^m)\). This exponential time requirement makes the sum-check protocol impractical for even relatively simple problems, as P cannot feasibly perform the necessary computations.

\section{General Purpose Interactive Proof Protocol}
\label{sec:org9ed87d4}

\subsection{Introduction to the GKR Protocol}
\label{sec:orgc88a16f}

\subsubsection{Recall: The Notion of Interactive Proofs in 1980s}
\label{sec:org78a33f7}
Recall \hyperref[sec:orgf2c105b]{Interactive Proof Systems} in 80s\cite{DBLP:conf/stoc/GoldwasserMR85}, when the IP model first came out, it was only a theoretical model, which means no one cares about the runtime of the all-powerful P.
At that time, the idea\cite{yaelgkrv} was people want to see how expressive, which computation can P prove to a polynomial time \(V\) by using interactive proofs.
\subsubsection{Delegating Computation: Interactive Proofs for Muggles}
\label{sec:org9f9b9c5}

In the \cite{DBLP:conf/stoc/GoldwasserKR08} paper,
the author introduced a protocol that can be used to \textbf{effectively} for both P and V to prove/verify the correctness of \textbf{general purpose computation}.
More formally: For any question/language computable by a boolean circuit \(C\) with depth \(d\) and input length \(n\), the protocol can ensure:
\begin{itemize}
\item The costs to V grow linearly with the depth \(d\) and input size \(n\) of the circuit, and only logarithmically with size \(S\) of the circuit.

\item \textbf{P's running time is polynomial in the input size \(n\).}

\end{itemize}

\subsubsection{Blueprint of the Protocol}

\paragraph{Layered Circuit:}
The protocol divides circuit \(C\) into \(d\) phases, since for each phase \(i\), if its gates value is wrong, then some gates' value in phase \((i+1)\) must be wrong.
More specifically, some gates that are connected to the error gates in layer \(i+1\) must be incorrect.

\paragraph{Local Correctness:}
We run a local \textbf{sum-check protocol} at each phase/layer to ensure \textbf{local correctness} by defining a function \(V_i: F^{s_i} \to F\) (where \(s_i := log Si\), the bits to represent \# gates in layer \(i\)), for any gate \(g_i\) in that layer \(i\), running \(V_i(g_i)\) will give us the corresponding gate value of \(g_i\).
By running \(d\)-subprotocols of each layer, where each protocol show connections between layer \(i\) and its layer before \((i+1)\). (More specifically, if \(V_i\) is not correct, then w.h.p \(V_{i+1}\) is not correct).
\textbf{Eventually, it will be reduced to a claim about the input values (layer \(d\)), which are known to \(V\).}

\subsection{Notations of the GKR Protocol}
\label{sec:org97e91cb}
Assume without loss of generality, circuit \(C\) has \(d\) layers, which means that each gate belongs to a layer, and each gate in layer \(i\) is connected by neighbors (determined) only in layer \(i+1\).
In a nutshell, the goal of the GKR protocol introduces a layered approach to circuit verification, significantly reducing \(P\)'s runtime to polynomial time.

\noindent \textbf{Notations: }
\cite{l12}Convert C to a layered arithmetic circuit with fan-in 2 with layer \(d\), and only consists of gates of the form \textbf{ADD} and \textbf{MULT}.
fan-in 2 means each gates in the i-th layer takes inputs from two gates in the (i+1)-th layer. We denote the number of gates in layer \(i\) as \(S_i\), and let \(s_i\) to be the number of input elements of layer \(i\). (\(s_i = \log{S_i}\))
We define function \(V_i(z_i)\) at each layer \(i\) to return the value of that gate with index \(z_i\)\footnote{In practice, we often choose \(m = \frac{\log{S_i}}{\log{log{S_i}}}\) and \(H = \{0, 1, ..., \log{S-1}\}\) to ensure \(F^{s_i}\) is not super-polynomial in size, for simplicity, here we use binary value as input.}:

\[
V_i: H^{s_i} \to F \textbf{  where  } H = \{0,1\} \textbf{  and  } s_i = \log{S_i}
\]

\noindent We designate \( V_0 \) as corresponding to the output of the circuit, and \( V_d \) as corresponding to the input layer. To define \( V_i \) for \( 1 \leq i \leq d \), recall the function \( W \) from Section \ref{lde}. Consider layer \( i \) of the circuit \( C \) as a vector of \( S_i \) gates:
\[
g = (g_1, g_2, \ldots, g_{S_i}),
\]
where each \( g_j \) represents the value of the \( j \)-th gate in that layer.

\noindent Define the function \( V_i: H^{s_i} \to F\) as follows:
\[
V_i(z) =
\begin{cases}
g_{a(z)} & \text{if } a(z) \text{ is a valid gate in layer } i, \\
0 & \text{otherwise}.
\end{cases}
\]
\noindent Note that for every \(p \in H^{s_i}\), we can express \(V_i(p)\) in the circuit as:

\[
V_i(p) = \sum_{w_1, w_2 \in H^{s_{i+1}}} (add_i(p, w_1, w_2) \times V_{i+1}(w_1) + V_{i+1}(w_2)) +(mult_i(p, w_1, w_2) V_{i+1}(w_1) \times V_{i+1}(w_2) )
\]

\noindent where \(add_i (mult_i)\) takes one gate label \(p \in H^{s_{i}}\) of layer i and two gate labels \(w_1\), \(w_2 \in H^{s_{i+1}}\) in layer \(i+1\), and outputs 1 if and only if gate p is an addition (multiplication) gate that takes the output of gate \(w_1, w_2\) as input, and otherwise 0.
\begin{algorithm}
\captionsetup{labelfont=bf, name=Protocol}
\caption{Constructing Multilinear Extensions of \(V_i\) at Layer i}
\label{construct}
In the i-th phase \((1 \leq i \leq d)\): P runs a local protocol with V to argue for the correctness of \(V_i\).
To do this, a local \textbf{sum-check protocol} of layer i will be applied, however, recall that the sum-check protocol only works if the expression inside the sum is a low-degree (multi-variate) polynomial, so let's try to convert \(V_i(p)\) to corresponding \textbf{multilinear extensions} by leveraging the Lagrange basis polynomial introduced in section \ref{lagrange}, for \(z \in F^{s_i}\):
\[
\widetilde{V_i}(z) = \sum_{p \in H^{s_i}} \widetilde{B}(z, p) \times \widetilde{V_i}'(p)
\]
Where \(\widetilde{V_i}'(z)\): \(F^{s_0} \to F\) refers to the LDE of \(V_i(p):\)
\[
\widetilde{V_i}'(z) = \sum_{w_1, w_2 \in H^{s_{i+1}}} (\widetilde{add_i}(z, w_1, w_2) \times (\widetilde{V_{i+1}}(w_1) + \widetilde{V_{i+1}}(w_2))) + (\widetilde{mult_i}(z, w_1, w_2) \times \widetilde{V_{i+1}}(w_1) \times \widetilde{V_{i+1}}(w_2))
\]
 Let \(\widetilde{add_i}, \widetilde{mult_i}: F^{s_i + 2s_{i+1}} \to F \) be the LDEs of \(add_i \text{ and } mult_i\), respectively.
 By replacing \(\widetilde{V_i}'(p)\) to its definition, we can get for every \(z_i\) in \(F^{s_i}\), let \(f_{i,z_i}: H^{s_i + 2s_{i+1}} \to F \) where:
\[
f_{i, z_i}(p, w_1, w_2) = \widetilde{B}(z_i, p) \times [(\widetilde{add_i}(p, w_1, w_2) \times (\widetilde{V_{i+1}}(w_1) + \widetilde{V_{i+1}}(w_2))) + (\widetilde{mult_i}(p, w_1, w_2) \times \widetilde{V_{i+1}}(w_1) \times \widetilde{V_{i+1}}(w_2))]
\]

\noindent Then \(\widetilde{V_i}(z_i): F^{s_0} \to F\) can be expressed as:
\[
\widetilde{V_i}(z_i) = \sum_{p \in H^{s_i}} \sum_{w_1, w_2 \in H^{s_{i+1}}}f_{i,z_i}(p, w_1, w_2)
\]

\end{algorithm}

\newpage

\subsection{GKR Protocol at Layer 0}
To begin with, let's describe the local sum-check protocol at layer 0 (the output layer) to show the idea of the GKR protocol. Recall:
\begin{itemize}
    \item \textbf{Layer 0 (Output Layer):} Contains a single gate \(z_0\).
    \item \textbf{Layer 1:} Contains gates \(w_j\) in \(H^{s_{i+1}}\), each associated with their LDEs \(\widetilde{V_{i+1}}(w_j)\).
\end{itemize}
In layer 0, P and V's task is to reduce the claim of \(\widetilde{V_0}(z_0) = r_0\) to the claim of \(\widetilde{V_1}(z_1) = r_1\), where \(z_1 \in F^{s_1}\) is a random value selected by V in layer 1.
\begin{algorithm}
\captionsetup{labelfont=bf, name=Protocol}
\caption{Sum-check Protocol at Layer 0}
\begin{enumerate}
    \item \textbf{Initial Claim:} P claims that \(\widetilde{V_0}(z_0) = r_0\), according to the definition of \(\widetilde{V_0}\), \(r_0\) is the output of the circuit \(C\).
    \item \textbf{Round 1:} P computes a univariate polynomial \(g_1(w_1): F^{s_1} \to F\) defined as:
    \[
    g_1(w_1) = \sum_{w_2 \in H^{s_1}} f_{0,z_0} (z_0, w_1, w_2)
    \]
    Which encapsulates the partial sum over \(w_2\) for each fixed \(w_1\) in layer 1.
    
    \noindent P then sends \(g_1\) to V, V needs to ensure that the total sum over all \(w_1\) equals \(r_0\), this is done by verifying that:
    \[
    \sum_{w_1 \in H^{s_1}} g_1(w_1) \equiv r_0
    \]
    V then selects a random value \(z_1 \in F^{s_1}\) and sends to P.
    \item \textbf{Final Round:} P computes a univariate polynomial \(g_2(w_2): F^{s_1} \to F\) defined as:
    \[
    g_2(w_2) = f_{0,z_0} (z_0, z_1, w_2)
    \]
    P then sends \(g_1\) to V, V then picks a random value \(z_2 \in F^{s_1}\) and verify the correctness of \(g_2\) on her own: Let's recall the definition of \(f_{i,z_i}: F^{s_i + 2s_{i+1}} \to F \) then \(f_{0,z_0}(z_0, z_1, z_2)\) can be defined as:
    \[
    \widetilde{B}(z_0, z_0) \times [(\widetilde{add_i}(z_0, z_1, z_2) \times (\widetilde{V_{i+1}}(z_1) + \widetilde{V_{i+1}}(z_2))) + (\widetilde{mult_i}(z_0, z_1, z_2) \times \widetilde{V_{i+1}}(z_1) \times \widetilde{V_{i+1}}(z_2))]
    \]
    V can compute on her own \(\widetilde{add_i}\) and \(\widetilde{mult_i}\) on \((z_0, z_1, z_2) \in F^{s_0 + 2s_1}\) (as well as other \((p_0, w_1, w_2) \in H^{s_0 + 2s_1}\) to check the correctness of the circuit). So as \(\widetilde{B}\).

    The computational burden in this verificational task is computing \(\widetilde{V_1}(z_1)\) and \(\widetilde{V_1}(z_2)\), since they related to the value of gates in layer \(2, 3, ... d\).

So instead, in this protocol, P sends both these values \(r_{1,1} = \widetilde{V_1}(z_1)\) and \(r_{1,2} = \widetilde{V_1}(z_2)\) to V, and claim they are true. And then using the following interactive reduction Protocol \ref{reduction} to "reduce" two claims into a single claim \(r_1 = \widetilde{V_1}(z_1')\) used in layer 1.
    
\end{enumerate}
\end{algorithm}

\noindent Here we show the Protocol \ref{reduction} which encodes two claims into a single polynomial, which efficiently reduces the verification burden across multiple layers. The reduction protocol combines the claims \( r_{1,1} = \widetilde{V_1}(z_1) \) and \( r_{1,2} = \widetilde{V_1}(z_2) \) into a single polynomial \( f(r) \), verified by \( V \) through random challenge \( r \), and in the next round P is left to prove the correctness of \(f(r) = r_1 \equiv \widetilde{V_1}(\gamma(r))\), which is a single claim.

\begin{algorithm}[H]
\captionsetup{labelfont=bf, name=Protocol}
\caption{Reduction Protocol at Layer 0~\cite{10.1145/2699436}}
\label{reduction}
    \begin{enumerate}
        \item \textbf{Initial Claims:} At the end of layer 0, P claims that:

        \[
        r_{1,1} = \widetilde{V_1}(z_1) \text{ and } r_{1,2} = \widetilde{V_1}(z_2)
        \]
        Both \( P \) and \( V \) agree on two distinct fixed elements \( t_1, t_2 \in F \).

        \item \textbf{Define Linear Map:} 
        V defines the unique linear polynomial \( \gamma: F \to F^{s_1} \) such that:
        \[
        \gamma(t_1) = z_1 \quad \text{and} \quad \gamma(t_2) = z_2
        \]
        This linear map effectively creates a line connecting \( z_1 \) and \( z_2 \) in the field \( F^{s_1} \).
        \item \textbf{P Sends \(f\) to V}: 
        \( P \) constructs and sends univariate polynomial \( f(t): F \to F \) to V that passes two points \((t_1, r_{1,1})\) and \((t_2, r_{1,2})\) defined as:
        \[
        f(t) = \widetilde{V_1}(\gamma(t)) = \widetilde{V_1}(z_1) \text{ at } t = t_1 \text{ and } \widetilde{V_1}(z_2) \text{ at } t = t_2
        \]

        \item \textbf{V Checks \(f\): } V checks that \(f\) indeed pass those two claimed points:
        \[
        f(t_1) = r_{1,1} \quad \text{and} \quad f(t_2) = r_{1,2}
        \]
        \item \textbf{Random Challenge by V:} If \(f\) is also in correct degree, V can argue the correctness of \(f(t) = \widetilde{V_1}(\gamma(t)) \) by choosing a random element \(r \in F\) and sends it to P.
        \item \textbf{Verification at Layer 1: } P and V sets:
        \[
        z_1' := \gamma(r) \quad \text{and} \quad r_1 := f(r)
        \]
        and initiates a Sum-check protocol at layer 1 to verify the correctness of single claim:
        \[
        r_1 = \widetilde{V_1}(z_1')
        \]
    \end{enumerate}
\end{algorithm}

\paragraph{Soundness of the Reduction Protocol:} 
In our example at layer 0, the protocol proceeds to layer 1 where a new Sum-check protocol is initiated to verify the correctness of \(\widetilde{V_1}(z_1)\). And just like previous round, in layer 1, V examines the structure of the circuit to ensure that \(\widetilde{V_1}(z_1)\) is correctly formed computed based on the circuit's specifications.

Let us say if P attempts to deceive by constructing an incorrectly \(f'(t)\) that only satisfies \(f'(t_1) = r_{1,1} \text{ and } f'(t_2) = r_{1,2}\), it is highly unlikely that \(f'(t)\) will also satisfy \(f'(r) = \widetilde{V_1}(\gamma(r))\) for the randomly chosen \(r\). The reason is due to the random selection of \(r\) by V, which ensures that P cannot predict \(r\) in advance to tailor \(f'(r)\) accordingly, and according to Schwartz-Zippel Lemma~\cite{sch}, if \(f'(r)\) does not correctly represent \(\widetilde{V_1}(\gamma(r))\), the probability that it passes the check at a random \(r\) is negligible for large \(|F|\). In other words, if the value \(r_1\) that \(r_1 = f(r)\) is consistent according to the sum-check protocol on \(r_1 = \widetilde{V_1}(\gamma(r))\) at layer 1, since \(r\) is randomly selected by V, this means \(f(t)\) is correctly formed with definition (\(f(t) = \widetilde{V_1}(\gamma(t))\)). And since \(f\) passes two points \((t_1, r_{1,1})\) and \((t_2, r_{1,2})\), we can get that:
\[
r_{1,1} = f(t_1) = \widetilde{V_1}(z_1) \text{ and } r_{1,2} = f(t_2) = \widetilde{V_1}(z_2)
\]
are both correct with high probability.

\subsection{GKR Protocol at Layer \(d\) and the Final Check}
\label{sec:org42bca73}

In layer \(d\), which is very similar to previous phases. P wants to convince V that \(r_d = \widetilde{V_d}(z_d')\), and at the end of the protocol in this layer, P will send \(g(z): F^{s_d} \to F\) \textbf{which refers to the low-degree polynomial of the input}, and V can verify on her own. If all the input matches, this means function \(g\) is correctly formed, thus according to the Sum-check protocol \(\widetilde{V_d}\) is also correctly formed, and in the previous layer, \(\widetilde{V_{d-1}}\) is also valid etc. Thus, according to the \hyperref[sec:org7e01221]{Soundness of the Sum-Check Protocol}, especially \hyperref[sec:orgbc212aa]{Backward Reasoning}, we can get  \(V_0\) is correctly formed, which implies \(C(x) = r_0\) should be correct w.h.p.

\subsection{Analyzing \(P\)'s Runtime Complexity}

The traditional Sum-Check Protocol is a foundational component in interactive proof systems, enabling the verification of polynomial evaluations over exponentially large domains. Specifically, recall Section \ref{limitation}, in each of \( m \) rounds, \(P\) must compute a univariate polynomial \( s_i(x_i) \) by summing over all \( 2^{m-i} \) possible assignments of the remaining \( m-i \) variables. This results in a runtime complexity of \( \mathcal{O}(Sm \cdot 2^m) \) for \(P\), where \( S \) denotes the degree of the polynomial. Such exponential time requirements make the traditional Sum-Check Protocol impractical for large-scale computations.

In contrast, the \textbf{GKR Protocol} leverages the structured nature of layered arithmetic circuits to achieve improved efficiency. Each layer \( i \) of the circuit comprises \( S_i \) gates with a fixed fan-in of 2, performing either addition or multiplication. Recall Protocol \ref{construct}, the gate function \( \widetilde{V_i}(z_i) \) is defined as:
\[
\widetilde{V_i}(z_i) = \sum_{p \in H^{s_i}} \sum_{w_1, w_2 \in H^{s_{i+1}}} f_{i,z_i}(p, w_1, w_2)
\]
where \( H \) represents the set of possible gate indices, and \( f_{i,z_i} \) encapsulates the computation performed by each gate in layer \( i \). the GKR Sum-Check Protocol aggregates contributions gate-by-gate within each layer, which means considering all gates in layer \(i+1\), and this is feasible because the values of gates in layer \( i+1 \) are precomputed and stored.
As a result, \(P\)'s runtime at each layer \( i \) becomes \( \text{poly}(S_i) \), where \( S_i \) is the number of gates in that layer. This polynomial-time complexity starkly contrasts with the exponential runtime of the traditional Sum-Check Protocol, underscoring the practicality of the GKR Protocol for delegating and verifying large computations. By capitalizing on the structural properties of layered arithmetic circuits, the GKR Protocol ensures that \(P\) can efficiently perform verification tasks, maintaining scalability even as the size of the computation grows.

\bibliography{abbrev0,crypto,ref}

\end{document}